\documentclass[10pt]{article}
\usepackage[margin=1in]{geometry}
\usepackage[utf8]{inputenc}
\usepackage{amsmath,amsthm}

\usepackage{graphicx}

\usepackage[linesnumbered,ruled,vlined]{algorithm2e}
\usepackage{amsmath, amssymb, amsfonts}

\usepackage{subfig}
\usepackage{listings}
\usepackage{mathtools}
\usepackage{soul}
\usepackage{lipsum}
\usepackage{algorithmic}
\usepackage{multicol}

\usepackage[normalem]{ulem}
\usepackage{xcolor}

\usepackage{comment}

\newtheorem{theorem}{Theorem}
\newtheorem{lemma}{Lemma}
\newtheorem{definition}{Definition}
\newtheorem{corollary}{Corollary}
\newtheorem{observation}{Observation}

\begin{document}

\title{Construction of a Byzantine Linearizable SWMR Atomic Register from SWSR Atomic Registers 
}


\author{
Ajay D. Kshemkalyani\footnote{University of Illinois Chicago, USA. Email: ajay@uic.edu},
Manaswini Piduguralla\footnote{Indian Institute of Technology Hyderabad, India. Email: cs20resch11007@iith.ac.in},
Sathya Peri\footnote{Indian Institute of Technology Hyderabad, India. Email: sathya\_p@cs.iith.ac.in},
Anshuman Misra\footnote{University of Illinois Chicago, USA.  Email: amisra7@uic.edu}
}
\maketitle

\begin{abstract}
The SWMR atomic register is a fundamental building block in shared memory distributed systems and implementing it from SWSR atomic registers is an important problem. While this problem has been solved in crash-prone systems, it has received less attention in Byzantine systems. Recently, Hu and Toueg gave such an implementation of the SWMR register from SWSR registers. While their definition of register linearizability is consistent with the definition of Byzantine linearizability of a concurrent history of Cohen and Keidar, it has these drawbacks. (1) If the writer is Byzantine, the register is linearizable no matter what values the correct readers return. (2) It ignores values written consistently by a Byzantine writer. We need a stronger notion of a {\em correct write operation}. (3) It allows a value written to just one or a few readers' SWSR registers to be returned, thereby not validating the intention of the writer to write that value honestly. (4) Its notion of a  ``current'' value returned by a correct reader is not related to the most recent value written by a correct write operation of a Byzantine writer. 
We need a more up to date version of the value that can be returned by a correct reader. 
In this paper, we give a stronger definition of a Byzantine linearizable register that overcomes the above drawbacks. Then we give a construction of a Byzantine linearizable SWMR atomic register from SWSR registers that meets our stronger definition. The construction is correct when $n>3f$, where $n$ is the number of readers, $f$ is the maximum number of Byzantine readers, and the writer can also be Byzantine. The construction relies on a public-key infrastructure.
\end{abstract}



\section{Introduction}
\label{section:intro}

Implementing shared registers from weaker types of registers is a fundamental problem in distributed systems and has been extensively studied \cite{DBLP:journals/jacm/AttiyaBD95,DBLP:conf/podc/BurnsP87,DBLP:journals/jacm/HaldarV95,DBLP:conf/podc/IsraeliS92,DBLP:journals/toplas/Herlihy91,DBLP:journals/toplas/HerlihyW90,DBLP:journals/dc/Lamport86,DBLP:journals/dc/Lamport86a,DBLP:journals/toplas/Peterson83,DBLP:conf/focs/PetersonB87,DBLP:conf/podc/SinghAG87,DBLP:conf/podc/Newman-Wolfe87,DBLP:journals/ipl/Vidyasankar88,DBLP:journals/ipl/Vidyasankar91,DBLP:conf/focs/VitanyiA86}.  We consider the problem of implementing a single-writer multi-reader register (SWMR) from single-writer single-reader (SWSR) registers in a system with Byzantine processes. This SWMR register in a Byzantine setting is of great importance in recent research.
For example, Mostefaoui et al. \cite{DBLP:journals/mst/MostefaouiPRJ17} prove that in message-passing systems with Byzantine failures, there is a $f$-resilient implementation of a SWMR register if and only if $f<n/3$ processes are faulty, where $f$ is the number of Byzantine processes and $n$ is the total number of processes. It was the first to give the definition of a linearizable SWMR register in the presence of Byzantine processes and \cite{DBLP:conf/wdag/CohenK21} generalized it to objects of any type. Aguilera et al. \cite{DBLP:conf/podc/AguileraBGMZ19} use atomic SWMR registers to solve some agreement problems in hybrid systems subject to Byzantine process failures. Cohen and Keidar \cite{DBLP:conf/wdag/CohenK21} give $f$-resilient implementations of three objects -- asset transfer, reliable broadcast, atomic snapshots -- using atomic SWMR registers in systems with Byzantine failures where at most $f<n/2$ processes are faulty. Their implementations were based on their definition of Byzantine linearizability of a concurrent history.

In other related work, a SWMR register was built above a message-passing system where processes communicate using send/receive primitives with the constraint that 
$f<n/3$
\cite{DBLP:journals/jpdc/ImbsRRS16,DBLP:journals/mst/MostefaouiPRJ17}. These works do not use signatures. Unbounded history registers were required in \cite{DBLP:journals/jpdc/ImbsRRS16} whereas \cite{DBLP:journals/mst/MostefaouiPRJ17} used $O(n^2)$ messages per write operation. Although building SWMR registers over SWSR registers or over message-passing systems is equivalent as SWSR registers can be emulated over send/receive and vice versa, this is a round-about and expensive solution. A similar problem for the client-server paradigm in message-passing systems was solved in \cite{DBLP:conf/srds/MalkhiR98} using cryptography. 

\subsection{Motivation}
\label{section:motivation}
The SWMR atomic register is seen to be a basic building block in shared memory distributed systems and implementing it from SWSR atomic registers is an important problem. While this problem has been solved in crash-prone systems, it has received recent attention in Byzantine systems. Recently, Hu and Toueg gave such an implementation of the SWMR register from SWSR registers \cite{DBLP:conf/wdag/0009T22}. While their definition of register linearizability is consistent with the definition of Byzantine linearizability of a concurrent history of Cohen and Keidar \cite{DBLP:conf/wdag/CohenK21}, both \cite{DBLP:conf/wdag/0009T22,DBLP:conf/wdag/CohenK21} as well as \cite{DBLP:journals/jpdc/ImbsRRS16,DBLP:journals/mst/MostefaouiPRJ17} have the following drawbacks. 
\begin{enumerate}
    \item If the writer is Byzantine, the register is vacuously linearizable no matter what values the correct readers return. Reads by correct processes can return any value whatsoever including the initial value while the register meets their definition of linearizability. In particular, there is no view consistency. For example, in the Hu-Toueg algorithm, consider a scenario where a Byzantine writer writes a different data value associated with the same counter value to the various readers' SWSR registers. The correct readers will return different data values associated with the same counter value, thus having inconsistent views. An example application where this is a problem is collaborative editing for a document hosted on a single server. Another reason why this is problematic is that it violates the agreement clause of the well-known consensus/Byzantine agreement problem, which requires that all non-faulty processes must agree on the same value even if the source is Byzantine. We require view consistency.
    \item Their definition of register linearizability does not factor in, or ignores, those values written by a Byzantine writer, by honestly following the writer protocol for those values. We need a stronger notion of a {\em correct write operation} that factors in such values as being written correctly. Also, note that the Byzantine writer is in control of the execution both above and below the SWMR register interface and hence the value that it writes in a correct write operation can be assumed to be the value intended to be written (correctly) and not altered by Byzantine behavior.
    \item Their definition of register linearizability allows a value written by a Byzantine writer to just a single reader's SWSR register to be returned by a correct process. In order to validate that the writer intended to write that value honestly, we would like a minimum threshold number of readers' SWSR registers to be written that same value to enable that value to become eligible for being returned to a correct reader. This validates the intention of the Byzantine writer to write that particular value. 
    \item In their definition of register linearizability, their notion of a ``current'' value returned by a correct reader is not related to the most recent value written by a correct write operation of a Byzantine writer. 
    We need a more up to date version of the value that can be returned by a correct reader. This helps give a stronger guarantee of progress from the readers' perspective.
\end{enumerate}
Our definition of a Byzantine linearizable register is stronger than not just that of \cite{DBLP:conf/wdag/CohenK21,DBLP:conf/wdag/0009T22} but also that of \cite{DBLP:journals/jpdc/ImbsRRS16,DBLP:journals/mst/MostefaouiPRJ17,DBLP:conf/srds/MalkhiR98} and overcomes the above drawbacks. Further, we are interested in implementing the SWMR register over SWSR registers directly in the shared memory model.

\subsection{Contributions}
\begin{enumerate}
\item In this paper, we give a stronger definition of a {\em Byzantine linearizable register} that overcomes all the above drawbacks of \cite{DBLP:conf/wdag/CohenK21,DBLP:conf/wdag/0009T22} and \cite{DBLP:journals/jpdc/ImbsRRS16,DBLP:journals/mst/MostefaouiPRJ17,DBLP:conf/srds/MalkhiR98}.
We introduce the concept of a {\em correct write operation} by a Byzantine writer as one that conforms to the write protocol. We also introduce the notion of a {\em pseudo-correct write operation} by a Byzantine writer, which has the effect of a correct write operation. 
Only correct and pseudo-correct writes may be returned by correct readers. The correct and pseudo-correct writes are totally ordered 
and this order 
is the total order in logical time \cite{Mattern88virtualtime,DBLP:journals/computer/RaynalS96}
in which the writes are performed. 
\item Then we give a construction of a Byzantine linearizable SWMR atomic register from SWSR atomic registers that meets our stronger definition. The construction is correct when $n>3f$, where $n$ is the number of readers, $f$ is the maximum number of Byzantine readers, and the writer can also be Byzantine. The construction relies on a public-key infrastructure (PKI). 

The construction develops the idea of the readers issuing logical read timestamps to the values set aside for them by the writer. Logical global states on the readers' SWSR registers, akin to consistent cuts in message-passing systems \cite{Mattern88virtualtime}, are then constructed. The algorithm logic ensures that values read at/along such global states form a total order, thereby helping to ensure Byzantine register linearizability.

As compared to the algorithm in \cite{DBLP:conf/wdag/0009T22} which can tolerate any number of Byzantine readers, our algorithm requires $f<n/3$. Also, in the algorithm in \cite{DBLP:conf/wdag/0009T22}, a reader that stops reading also stops taking implementation steps whereas our algorithm requires a reader helper thread to take infinitely many steps even if it has no read operation to apply. The algorithm in \cite{DBLP:conf/wdag/0009T22} as well as our algorithm use a PKI.
\end{enumerate}

\noindent{{\bf Outline:}} Section~\ref{section:modelprelim} gives the system model and preliminaries. Section~\ref{section:characterization} gives our characterization of a Byzantine linearizable register based on logical time and culminates in the definition of such a register. Section~\ref{section:algorithm} gives our construction of the SWMR Byzantine linearizable register using SWSR registers. Section~\ref{section:proof} gives the correctness proof. Section~\ref{section:conc} concludes.

\section{Model and Preliminaries}
\label{section:modelprelim}
\subsection{Model Basics}
We consider the shared memory model of a distributed system. The system contains a set $P$ of asynchronous processes. These processes access some shared memory objects. All inter-process communication is done through an API exposed by the objects. Processes invoke operations that return some response to the invoking process. We assume reliable shared memory but allow for an adversary to corrupt up to $f$ processes in the course of a run. A corrupted process is defined as being {\em Byzantine} and such a process may deviate arbitrarily from the protocol. A non-Byzantine process is {\em correct} and such a process follows the protocol and takes infinitely many steps.

We also assume a 
PKI. 
Using this, each process has a public-private key pair used to sign data and verify signatures of other processes. A values $v$ signed by process $p$ is denoted $\langle v\rangle_p$.

We give an algorithm that emulates an object $O$, viz., a SWMR register from SWSR registers. We assume that there is adequate access control such that a SWSR register can be accessed only by the single writer and the single reader between whom the register is set up, and that another (Byzantine) process cannot access it. The algorithm is organized as methods of $O$. A method execution is a sequence of steps. It begins with the {\em invoke} step, goes through steps that access lower-level objects, viz., SWSR registers, and ends with a {\em return} step. The invocation and response delineate the method's execution interval. In an {\em execution} $\sigma$, each correct process invokes methods sequentially, and steps of different processes are interleaved. Byzantine processes take arbitrary steps irrespective of the protocol. The {\em history} $H$ of an execution $\sigma$ is the sequence of high-level invocation and response events of the emulated SWMR register in $\sigma$.
A history $H$ defines a partial order $\prec_H$ on operations. $op_1\prec_H op_2$ if the response event of $op_1$ precedes the invocation event of $op_2$ in $H$. $op_1$ is concurrent with $op_2$ if neither precedes the other.

In our algorithm, we assume that each reader process has a helper thread that takes infinitely many steps even if the reader stops reading the implemented register. These steps are outside the invocation-response intervals of the readers' own operations. Also, the linearization point of a pseudo-correct write operation may fall after the invocation-response interval. These are non-standard features of our shared memory model.

\subsection{Linearizability of a Concurrent History}
{\em Linearizability}, a popular correctness condition for concurrent objects, is defined using an object's sequential specification. 
\begin{definition}
\label{definition:linearization}
(Linearization of a concurrent history:) A {\em linearization} of a concurrent history $H$ of object $o$ is a sequential history $H'$ such that:
\begin{enumerate}
    \item After removing some pending operations from $H$ and completing others by adding matching responses, it contains the same invocations and responses as $H'$,
    \item $H'$ preserves the partial order $\prec_H$, and
    \item $H'$ satisfies $o$'s sequential specification.
\end{enumerate}
\end{definition}

A SWMR register as well as a SWSR register expose the {\em read} and {\em write} operations. The sequential specification of a SWMR and a SWSR register states that a read operation from register $Reg$ returns the value last written to $Reg$. Following Cohen and Keidar \cite{DBLP:conf/wdag/CohenK21}, we manage Byzantine behavior in a way that provides consistency to correct processes. This is achieved by linearizing correct processes' operations and offering a degree of freedom to embed additional operations by Byzantine processes.

Let $H|_{correct}$ denote the projection of history $H$ to all correct processes. History $H$ is Byzantine linearizable if $H|_{correct}$ can be augmented by (some) operations of Byzantine processes such that the completed history is linearizable. Thus, there is another history with the same operations by correct processes as in $H$, and additional operations by at most $f$ Byzantine processes. 
\begin{definition}
\label{definition:blin}
(Byzantine linearization of a concurrent history \cite{DBLP:conf/wdag/CohenK21}:) A history $H$ is Byzantine linearizable if there exists a history $H'$ such that $H'|_{correct} = H|_{correct}$ and $H'$ is linearizable.
\end{definition}
An object supports Byzantine linearizable executions if all of its executions are Byzantine linearizable. SWMR 
registers support Byzantine linearizable executions because before every read from such a register, invoked by a correct process, one can add a corresponding Byzantine write. 

\subsection{Linearizability of Register Implementations}
\label{section:linri}
Hu and Toueg defined register linearizability in a system with Byzantine processes as follows \cite{DBLP:conf/wdag/0009T22}. They let $v_0$ be the initial value of the implemented register and $v_k$ be the value written by the  $k$th write operation by the writer $w$ of the implemented register.
\begin{definition}
\label{definition:HuToueg}
({\em Register Linearizability \cite{DBLP:conf/wdag/0009T22}:}) In a system with Byzantine process failures, an implementation of a SWMR register is linearizable if and only if the following holds. If the writer is not malicious, then:
\begin{itemize}
\item (Reading a ``current'' value) If a read operation {\tt R} by a process that is not malicious returns the value $v$ then:
\begin{itemize}
    \item there is a write $v$ operation that immediately precedes {\tt R} or is concurrent with {\tt R}, or
    \item $v=v_0$ (the initial value) and no write operation precedes {\tt R}.
\end{itemize}
\item (No ``new-old'' inversion) If two read operations {\tt R} and {\tt R'} by processes that are not malicious return values $v_k$ and $v_{k'}$, respectively, and {\tt R} precedes {\tt R'}, then $k\leq k'$.
\end{itemize}
\end{definition}
While this definition of register linearizability is consistent with the definition of a Byzantine linearization of a concurrent history (Definition~\ref{definition:blin}), in the sense that both are concerned only with correct processes' views, it is not ideal for the reasons given in Section~\ref{section:motivation}. 
Therefore the register should meet stronger criteria of a linearizable register, in the face of Byzantine processes, to accommodate the behavior of the Byzantine writer when it is behaving (writing) correctly.
We term such a register as a {\em Byzantine linearizable register}. In this paper, we first define a Byzantine linearizable register, and then solve the problem of constructing a Byzantine linearizable SWMR register from SWSR registers.

\section{Characterization of Byzantine Register Linearizability}
\label{section:characterization}
The object SWMR register supports Byzantine linearizable executions \cite{DBLP:conf/wdag/CohenK21}. However, we need to construct a SWMR register from SWSR registers. We assume wlog that there are $n$ SWSR registers $R\_init_{wi}$ writable by the single writer $w$ and readable by reader $i \in [i,n]$. Here we characterize the requirements for such a construction, culminating in Definition~\ref{definition:lint} of Byzantine Register Linearizability. 

The writer as well as the reader processes can be Byzantine. As a Byzantine reader can return any value whatsoever, the linearizability specification is based on values that correct readers return. The Byzantine writer can behave anyhow and can write different values to the SR registers, or write different values to different subsets of SR registers while not writing to some of them at all, 
or write multiple different values over time to the same some or all SR registers, as part of the same write operation. We assume that $t$ of the $n$ readers are Byzantine. In our characterization, we seek recourse to logical time but present the final definition of the Byzantine linearizable register using physical time.

We refer to a write operation $o(v)$ by a timestamp vector $T$  of size $n$, where $T[i]$ is the logical timestamp assigned 
to a value of $v$ written into $R\_init_{wi}$. 
This vector gives the {\em logical time vector} of the write operation $o(v)$. As $T[i]$ is a logical timestamp, it can equivalently be assigned by reader $i$. Thus, correct readers assign monotonically increasing timestamp values whereas a timestamp that violates this must have been assigned by a Byzantine reader and can be rejected/ignored by correct readers. We define the relations $\leq$, $<$, and $||$ (concurrent) on the set of timestamp vectors ${\cal T}$ in the standard way as follows.
\begin{itemize}
\item $T_1\leq T_2 =_{def} \forall i \in[1,n], T_1[i]\leq T_2[i]$
\item $T_1<T_2 =_{def} T_1\leq T_2 \wedge \exists i\,|\, T_1[i]<T_2[i]$
\item $T_1\,||\,T_2 =_{def} \exists i \,| T_1[i]<T_2[i] \wedge \exists j\,|\, T_1[j]>T_2[j]$
\end{itemize}

$({\cal T},<)$ forms a lattice.

In general, when an object $O_1$, denoted a {\em high-level object (HLO)} is simulated or constructed using objects of another type $O_2$, denoted a {\em low-level object (LLO)}, there are two interfaces. A process interacts with the HLO through a {\em high-level interface (HLI)} through alternating invocations and matching responses. Between such a pair of matching invocation and response, the process interacts with the LLO through a {\em low-level interface (LLI)} using alternating invocations and responses. Such interactions are in software. 

For our problem, the HLO is the Byzantine-tolerant SWMR atomic register and the HLI is the read and write operation. The LLO is the SWSR atomic register and the LLI is also the read and write operation. We term the program code executed below the HLI and above the LLI for a single invocation of a write/read at the HLI as the code or protocol for the (HLI) write operation/read operation, respectively. 

In the face of Byzantine readers as well as a Byzantine writer, we need to define a correct write operation. In the sequel, we use $u$ or $v$ to refer to the actual data value written. A $write(v)$ invocation at the HLI may be converted at a Byzantine writer into possibly multiple operation invocations for different $write(v')$ at the LLI to all or some subset of the various instances of the LLO. If a $write(v)$ invocation at the HLI is converted by a Byzantine writer into an invocation of $write(v')$ and it executes the protocol exactly for this value $v'$, it is considered as a correct write operation because that can be taken to be the value the writer writes or intended to write. Likewise if the $write(v)$ at the HLI is converted into multiple serial invocations of $write(v')$ (for different values of $v'$) and the protocol for each of these $v'$ is correctly followed, these various $write(v')$ are considered correct write operations because that sequence of write operations can be taken to be the values the writer writes or intended to write. This is because the invocation/response at the HLI is at a Byzantine process which controls the execution of code above the LLI and above the HLI. In a correct write operation, the code between the HLI and the LLI is followed correctly by the Byzantine process.

\begin{definition}
\label{definition:cwot}
A {\em correct write operation} is a write operation that follows the write protocol.
\end{definition}

So far in the literature \cite{DBLP:conf/wdag/CohenK21,DBLP:conf/wdag/0009T22}, any behavior of a Byzantine writer is allowable in the linearizability definition. We accommodate a Byzantine writer differently and introduce the concept of a {\em pseudo-correct write operation} (Definition~\ref{definition:pcwot}), which is a Byzantine write operation that has the effect of a correct write operation. This is first informally motivated as follows. A Byzantine write operation can, for example, 
\begin{enumerate}
\item 
write multiple values to the various $R\_init_{wi}$ (possibly resulting in multiple pseudo-correct write operations) or
\item together with earlier write operations write a single value (possibly resulting in a pseudo-correct write operation), or 
\item together with earlier write operations that wrote different values write those values (possibly resulting in multiple pseudo-correct write operations).
\end{enumerate}
Thus, there is no longer a one-one mapping from write operations issued to the HLI object interface to values written to the object; it is a many-many mapping. {\em If a pseudo-correct write operation is a result of values written across multiple HLO write operations, we define that pseudo-correct write operation to occur in the latest of those HLO write operations. Its invocation and response are those of that latest HLO operation.
Its linearization point occurs when the write takes effect for correct readers and this can happen even after the HLO write response to the HLO write invocation in which the pseudo-correct write occurred.}
Due to the many-many mapping, one HLO write invocation-response can result in different values being written by the pseudo-correct operations and read/returned by correct readers. This does not pose any ambiguity because the different values that are returned to the correct readers have different logical timestamp vectors.

\begin{definition}
\label{definition:ppcwot}
A {\em potential pseudo-correct write operation} of value $v$ is a write operation, timestamped $T(v)$, that may not follow the write protocol but 
\begin{enumerate}
\item 
$T$ is such that there does not exist any correct write operation timestamp $T'$ where $\exists i\,|\, T[i]<T'[i] \wedge \exists j\,|\, T[j]>T'[j]$, and
\item there is a quorum of size $\geq n-t$ indices $i$ such that $v$ was written to $R\_init_{wi}$ and logically timestamped $T(v)[i]$ (equivalently and in practice by reader $i$). 
\end{enumerate}
\end{definition}

\begin{definition}
\label{definition:ret}
A write operation {\em stabilizes} if its value is returnable, meaning eligible for being returned, by a correct reader.
\end{definition}
A correct write operation always stabilizes whereas a potential pseudo-correct write may stabilize, depending on run-time dynamic data races, steps of Byzantine readers, and the algorithm. Only all write operations that stabilize have a linearization point.

\begin{definition}
\label{definition:pcwot}
A {\em pseudo-correct write operation} is a potential pseudo-correct write operation that stabilizes.
\end{definition} 

\begin{definition}
\label{definition:monotonicity}
({\em Monotonicity/Total Order of stabilized write operation vector timestamps Property}:) The set of write operation timestamps that stabilize is totally ordered.
\footnote{This definition is not about the monotonicity of vector timestamp values as a function of the physical time order in which their operations stabilize, but about the set of such vector timestamps which is totally ordered.}  
\end{definition}
If the Monotonicity/Total Order Property is satisfied, of any two potential pseudo-correct writes whose vector timestamps are concurrent, at most one can stabilize.

\begin{definition}
\label{definition:advancetv}
({\em Genuine Advance of Timestamp Property:})
%
For vector timestamps $T_1(v_1)$ and $T_2(v_2)$ of correct and pseudo-correct write operations 
such that $T_1<T_2$, there is index $i$ of a correct reader process such that $T_1[i]<T_2[i]$. 
\end{definition}
In conjunction with the Monotonicity/Total Order Property, the Genuine Advance of Timestamp Property guarantees that there has been progress from $T_1$ to $T_2$ and this progress includes a new write of $v_2$ to a $R\_init_{wi}$ register for a correct process $i$. Thus the progress is not a fake operation reported by a Byzantine reader. 
Further, the latest values as per $T_2$ at $n-t$ $R\_init_{wj}$ are also $v_2$. 

Let $T_{corr}$ denote the timestamp of a correct or a pseudo-correct write operation, which is the timestamp of a stabilized write operation. 
\begin{definition}
\label{definition:consistentt}
A {\em consistent timestamp} of a write operation 
is a vector timestamp $T$ such that $\not\exists T_{corr}\,|\,(\exists i\,|\, T[i]<T_{corr}[i] \wedge \exists j\,|\, T[j]>T_{corr}[j])$.
\end{definition}
The set of consistent timestamps forms a sublattice $({\cal T}_c,<)$ of $({\cal T},<)$.  
No consistent write timestamp is concurrent with 
a correct or pseudo-correct write timestamp. Thus at run-time, it is determined that a correct write timestamp and a pseudo-correct write timestamp is an {\em inevitable timestamp} in $({\cal T}_c,<)$. The execution path traced in the lattice $({\cal T}_c,<)$ passes through these correct and pseudo-correct write operation timestamp states.

\begin{definition}
\label{definition:view}
({\em View Consistency Property:}) If a write operation $T(v)$ is returned to a correct reader's read {\tt R}, it is returned to all correct readers' read operations that are preceded by {\tt R}, assuming there are no further writes to $R\_init_{wi}$ after $T(v)$.
\end{definition}
\begin{definition}
\label{definition:to}
({\em Total Ordering Property:}) If two write operations $T(v_1)$ and $T(v_2)$ are seen by any two correct readers, they are seen in a same common order.
\end{definition}
View Stability and Total Ordering properties together guarantee that all correct readers can 
see all correct and pseudo-correct writes in the same order.

Only correct and pseudo-correct writes may be returned by correct readers. A correct reader cannot distinguish between a correct and a pseudo-correct write operation. A pseudo-correct write operation $o_1(v_1)$ timestamped $T_1(v_1)$ may lose a race due to asynchrony of process executions to a pseudo-correct or correct write operation timestamped $T_2(v_2)$ where $T_1<T_2$, in which case $v_1$ is not actually returned to any read operation and $o_1$ is deemed to have an {\em invisible} linearization point.  The correct and pseudo-correct writes are totally ordered by their linearization points, and (if the Monotonicity/Total Order on Vector Timestamps Property is satisfied) this order is 
(a) the total order on their timestamps, which is 
(b) the total order in logical time in which these writes were performed, as also 
(c) the total order in which these write operation timestamped values are encountered in a run-time traversal of the lattice $({\cal T}_c,<)$ and potentially returned by 
HLI read operations. 
\footnote{Note this this total order on the linearization points is not uniquely defined by the natural total order in which operations were issued to the HLI object interface, because as observed earlier, there is a many-many mapping from the writes issued to the HLI object interface to values written to the object.}

In our characterization, we used logical time and LLOs but now present the final definition of the Byzantine linearizable register using physical time and HLOs.
Let $v^i$ be the value written by the $i$th correct or pseudo-correct write $W^i$, following the notation in \cite{amp}. Note that to determine $i$, $v^i$ and $W^i$ requires knowing what happened below the HLI and above the LLI because of the nature of pseudo-correct writes; but there is actually no need to determine $i$, $v^i$, and $W^i$.

\begin{definition}
\label{definition:lint}
(Byzantine Linearizable Register). In a system with Byzantine process failures, an implementation of a SWMR register is linearizable if and only if the following two properties are satisfied.
\begin{enumerate}
\item {\bf Reading a current value:} When a read operation {\tt R} by a non-Byzantine process returns the value $v$:
\begin{enumerate}
    \item if $v=v_0$ then no correct or pseudo-correct write operation precedes {\tt R}
    \item \label{toughcase} else if $v\neq v_0$ then $v$ was written by the most recent correct write operation that precedes {\tt R} or by a later pseudo-correct or correct write operation (either a pseudo-correct write operation, that precedes or overlaps with {\tt R}, or a correct write operation that overlaps {\tt R}).
\end{enumerate}
\item {\bf No ``new-old'' inversions:} If read operations {\tt R} and {\tt R'} by non-Byzantine processes return values $v^i$ and $v^j$, respectively, 
and {\tt R} precedes {\tt R'}, then $i\leq j$. 
\end{enumerate}
\end{definition}
In Definition~\ref{definition:lint}, ``precedes'', ``overlaps'', and ``later'' are with respect to physical time of HLI invocations and responses. In Case~\ref{toughcase}, note that the most recent correct write operation {\tt W} that precedes {\tt R} has its linearization point within the physical time duration of {\tt W}. Later pseudo-correct write operations that precede {\tt R} or overlap with {\tt R} may have their linearization points after their response in physical time. In particular, the most recent pseudo-correct write operation that precedes {\tt R} in physical time may have its linearization point after its physical time duration completes and hence it may be during {\tt R} or even after {\tt R} completes. 

Additionally, Monotonicity/Total Order of Vector Timestamps Property, Genuine Advance Property, View Consistency Property, and Total Ordering Property are useful properties that overcome the drawbacks of the previous definition(s) of Byzantine Register Linearizability discussed in Section~\ref{section:motivation}.

\section{The Algorithm}
\label{section:algorithm}
\subsection{Basic Idea and Operation}
Because the writer $w$ is Byzantine, a reader $i$ cannot unilaterally use the value in its register $R\_init_{wi}$ but needs to coordinate with other readers. But a reader may not invoke a read operation indefinitely. So we assume that each reader process has a reader helper thread that is always running and participates in this coordination. A correct write operation needs to wait for acknowledgements from the readers so that the write value is guaranteed to get written in the algorithm data structures and stabilize, and is not overwritten before it stabilizes. This guarantees progress by ensuring that the most recent correct write operation advances with time. The algorithm data structures are described in Algorithm~\ref{alg:lin}. The reader helper thread is given in Algorithm~\ref{alg:helperlin}. $w$ is the writer process and $P$ denotes the set of $n$ readers.

The writer writes $\langle k,u\rangle$, where $u$ is the value to be written and $k$ is a sequence number assigned by the writer, to all the reader registers $R\_init_{wi}$ and then waits for $n-t$ of the acknowledgement registers $R\_ack_{iw}$ to be written with this value.

The reader helper thread ($\forall p\in P$) loops forever. In each iteration, it reads $R\_init_{wp}$ and if the value overwrites the earlier value, it increments its local logical time $s$ (the {\em witness time}) and writes $\langle\langle k,u\rangle, s,p\rangle$ to all the registers $R\_witness_{pi}$. It then reads each $R\_witness_{ip}$ and if the logical time of an entry is larger than the previous logical time read, it stores the value in a local variable $T\_witness_i$ (if less than, or equal but the entry is different, $i$ is marked as Byzantine).
If there are at least $n-t$ $T\_witness_q$ entries with identical $\langle k,u\rangle$ values then  these $T\_witness_q$ entries are placed in the local set $Witness\_Set$. This $Witness\_Set$ signed by $p$ is written to all $R\_inform_{pi}$.
Each $R\_inform_{ip}$ is read to local variable $T\_inform_i$. If there are at least $n-t$ $T\_inform_j$ having identical entries $\langle\langle k,u\rangle, s_l,l\rangle$ for at least $n-t$ readers $l$, these $T\_inform_j$ are placed in local variable $Inform\_Set$, which is then written in all $R\_final_{pi}$, and $R\_ack_{pw}$ is updated with $\langle k,u\rangle$. 
All the $R\_final_{ip}$ are read and placed in set $Z$. 
All the elements in $Z$, i.e., ($\forall i$) $R\_final_{ip}$ are totally ordered by a relation $\mapsto$, (to be defined in Definition~\ref{definition:mapsto}),  as will be proved in Theorem~\ref{theorem:notconc}. The latest element in $Z$ is identified as $Y$ via a call to $Find\_Latest(Z)$ and if $Y$ is different from the local $Inform\_set$, (a) it is written to all $R\_final_{pi}$ and (b) $R\_ack_{pw}$ is updated with the common $\langle k,u\rangle$ occurring in identical $T\_witness$ entries $\langle\langle k,u\rangle, s_l,l\rangle$ for $\geq$ $n-t$ values of $l$ in all the at least $n-t$ $Witness\_Sets$ in $Inform\_Set$ that is $Y$.

Observe that $R\_inform_{**}$ registers are needed to ensure that the $R\_final_{**}$ entries, i.e., their vector timestamps, form a total order; otherwise if two correct processes were to concurrently write their respective $Witness\_Set$s into their $R\_final$ rows, the two entries may be partially ordered. This introduces an additional level of indirection in the algorithm.

\begin{algorithm}[th!]
{\small 
\SetKwComment{Comment}{$\triangleright$ }{}
{\bf Shared variables:\\}
For all processes $i\in P$: $R\_init_{wi}$: atomic SWSR register  $\leftarrow \langle 0,u_0\rangle$\\
For all processes $i\in P$: $R\_ack_{iw}$: atomic SWSR register  $\leftarrow \langle 0,u_0\rangle$\\
For all processes $i$ and $j$ in $P$:
$R\_witness_{ij}$: atomic SWSR register  $\leftarrow \langle\langle 0,u_0\rangle, 0,i\rangle$\\
For all processes $i$ and $j$ in $P$: $R\_inform_{ij}$: atomic SWSR register  $\leftarrow \langle\bigcup_{k\in P}\{ \langle\langle 0,u_0\rangle, 0,k \rangle\}\rangle_i$\\
For all processes $i$ and $j$ in $P$: $R\_final_{ij}$: atomic SWSR register  $\leftarrow \bigcup_{l\in L}\{\langle\bigcup_{k\in P}\{ \langle\langle 0,u_0\rangle, 0,k \rangle\}\rangle_l\}$, where $|L|\geq n-t \wedge L\subseteq P$\\
\medskip
{\bf Local variables:}\\
variable of $w$: $c$ $\leftarrow$ 0\\
variable of $i \in P$: $s$ $\leftarrow$ 0\\
variables of $j\in P$: For all processes $i\in P$: $T\_witness_i$ $\leftarrow$ $\langle\langle 0,u_0\rangle, 0,i\rangle$\\
variables of $j\in P$: For all processes $i\in P$: $T\_inform_i$ $\leftarrow$ $\langle\bigcup_{k\in P}\{ \langle\langle 0,u_0\rangle, 0,k\rangle\}\rangle_i$\\
variable of $j\in P$: $Witness\_Set$ $\leftarrow$ $\langle\bigcup_{k\in P}\{ \langle\langle 0,u_0\rangle, 0,k\rangle\}\rangle_p$\\ 
variable of $j\in P$: $Inform\_Set$ $\leftarrow$ $\bigcup_{l\in L}\{\langle\bigcup_{k\in P}\{ \langle\langle 0,u_0\rangle, 0,k\rangle\}\rangle_l\}$, where $|L|\geq n-t \wedge L\subseteq P$\\
\medskip

\underline{WRITE($u$):}\\
$c=c+1$\\
$W(\langle c,u\rangle)$\\
{\bf return}\\

\medskip
\underline{READ():}\\
$R()$\\
{\bf return} the value returned by the $R$ call\\

\medskip
\underline{$W(\langle k,u\rangle)$:} \\
\For{{\bf every} process $i\in P$}{
 $R\_init_{wi}=\langle k,u\rangle$
}
$d=0$\\
\While{$d< n-t$}{
 \For{each new $\langle k,u\rangle$ in $R\_ack_{iw}$ among $R\_ack$ registers}{
  $d=d+1$
 }
}
{\bf return done}

\medskip
\underline{$R()$:}\\
execute one iteration of the reader helper thread, Algorithm~\ref{alg:helperlin}\\
{\bf return} the value last written in $R\_ack_{pw}$ in that execution

}
\caption{Constructing a linearizable SWMR atomic register. Code at process $p$.}
\label{alg:lin}
\end{algorithm}

\begin{algorithm}[th!]
{\small 
\SetKwComment{Comment}{$\triangleright$ }{}
\medskip
\underline{reader helper thread:} \\
\While{$true$}{
 \If{$R\_init_{wp} (=\langle k,u\rangle)$ is newly written (overwrites a different value) 
 }{
   $s=s+1$\\
   \For{each $i\in P$}{
    $R\_witness_{pi}\leftarrow \langle\langle k,u\rangle, s,p\rangle$
   }
  }
  \For{each $i\in P$}{
  \If{$R\_witness_{ip} (=\langle\langle k,u\rangle, s',i\rangle)$ is newly written, i.e., $s' > T\_witness_i.s$ 
  \label{line:ku}
  }{
   $T\_witness_i\leftarrow \langle\langle k,u\rangle, s',i\rangle$
  }
  \ElseIf{$s'<T\_witness_i.s \vee (s'=T\_witness_i.s \wedge \langle k,u\rangle (Line~\ref{line:ku}) \neq T\_witness_i.\langle k,u\rangle)$}{
   (optionally) mark process $i$ as Byzantine
  }
 }

 \If{$\exists$ $\geq$ $n-t$ latest elements $\langle\langle k,u\rangle,*,q\rangle$ $T\_witness_q$ 
 }{
  add all these $\geq n-t$ elements $T\_witness_q$ to the emptyset $Witness\_Set$\\ 
  \For{each $i\in P$}{
   $R\_inform_{pi}\leftarrow \langle Witness\_Set\rangle_p$
  }
  \For{each $i\in P$}{
   $T\_inform_{i}\leftarrow R\_inform_{ip}$
  }
  \If{$\exists \geq n-t$ processes $j$ (i.e., $T\_inform_j$) having identical entries $\langle\langle k,u\rangle,s_l,l\rangle$ for $\geq n-t$ processes $l$}{
   Place the $\geq n-t$ $T\_inform_j$ in emptyset $Inform\_Set$\\
   \For{each $i\in P$}{
    $R\_final_{pi}\leftarrow Inform\_Set$ 
   }
   $R\_ack_{pw}\leftarrow \langle k,u\rangle$
  }
 }
 $Z\leftarrow\emptyset$\\
 \For{each $i\in P$}{
  $Z\leftarrow Z \cup \{R\_final_{ip}\}$\\
 }
 $Y\leftarrow Find\_Latest(Z)$\\
 \If{$Y\neq Inform\_Set$}{
  \For{each $i\in P$}{
   $R\_final_{pi}\leftarrow Y$
  }
  $R\_ack_{pw}\leftarrow$ the common $\langle k,u\rangle$ value occurring in identical $T\_witness$ entries $\langle\langle k,u\rangle,s_l,l\rangle$ for $\geq n-t$ values of $l$ in
  all the $\geq n-t$ $Witness\_Set$s in $Inform\_Set$ that is $Y$ 
 }
}
\medskip
\underline{$Find\_Latest(Z)$:}\\
\While{$|Z|>1$}{
 let $Z_i, Z_j$ be any two elements of $Z$; $Z_i$ ($Z_j$) has $\geq n-t$ entries and is the $Inform\_Set$ of some reader\\
 Let $Q_i$ ($Q_j$) be the set of $\geq n-t$ processes that have identical $T\_witness$ entries in the $\geq n-t$ $Witness\_Set$s in $Inform\_Set$ that is $Z_i$ ($Z_j$). For the $\geq n-2t$ readers $q\in Q_i\cap Q_j$, let $e^i_q = \langle\langle k,u\rangle, s^i_q,q\rangle$ and $e^j_q = \langle\langle k',u'\rangle, s^j_q,q\rangle$ be these $T\_witness$ entries in the $\geq n-t$ $Witness\_Set$s in the $Inform\_Set$s that are $Z_i$ and $Z_j$, resp.\\

 \If{$s^i_{q} \geq s^j_{q}$ for any/all 
 processes $q$}{
  $Z\leftarrow Z\setminus Z_j$
 }
 \Else{
  $Z\leftarrow Z\setminus Z_i$
 }
}
return(element in $Z$)

} 
\caption{Reader helper thread for constructing a linearizable SWMR atomic register. Code at process $p$.}
\label{alg:helperlin}
\end{algorithm}

\subsection{Instantiating Framework Definitions in the Algorithm}
A reader sees the value written by the most recent correct write operation that precedes the read operation, or it may return a later written value that is written by a correct or pseudo-correct write operation. In the context of our algorithm, the definitions of a correct write operation, a potential pseudo-correct write operation, and a pseudo-correct write operation apply directly. 
A write operation {\em stabilizes} (Definition~\ref{definition:stabilization}) if the value is potentially returnable by a correct read operation, i.e., when it is written to $R\_final_{p*}$ for all values of $*$ (in the form of an $Inform\_Set$ containing the required number of correctly signed $Witness\_Set$s).
A correct write operation always stabilizes whereas a potential pseudo-correct write operation may stabilize depending on the outcome of concurrency data races and behavior of the writer and Byzantine readers. 

A correct write operation always stabilizes because it waits for $n-t$ acknowledgements in $R\_ack_{*w}$ before completion, thereby allowing the value it has written to  $R\_final_{p*}$ to be eligible for being returned by a read operation. 
A correct write operation corresponds to a sufficient condition for stabilization. Writing the same value to $n-t$ readers' $R\_init_{wi}$ in a potential pseudo-correct operation is a necessary condition for stabilization. For the potential pseudo-correct write operation to become a pseudo-correct write operation, the Byzantine readers need to collaborate to allow the value being written to the other correct readers' $R\_init_{wi}$ to stabilize. 
Recall that a pseudo-correct write operation may have the value written to the readers' $R\_init_{wi}$ registers across multiple prior write operations.
The set of correct and pseudo-correct writes is exactly the set of writes whose values stabilize (follows from Theorem~\ref{theorem:onlystabilize}).

We show (Corollary~\ref{corollary:totalorder}) that the set of all values that stabilize are totally ordered by the logical times of their ``reading'' from the readers' registers $R\_init_{w*}$. 
Specifically, we use $\overline{T}$ to denote the vector of
logical times of reading a value $\langle k,u\rangle$ from $R\_init_{wi}$ by the various readers $i$, and we denote this total order relation $\mapsto$. The notation $\overline{T}$ as opposed to the timestamp vector $T$ we introduced in Section~\ref{section:characterization} is useful because not all readers may report the logical times of their reading $u$ from $R\_init_{w*}$. Thus $\overline{T}(u)$ may not have all $n$ components whereas $T(u)$ has all $n$ components. Roughly speaking, the $\overline{T_1}\mapsto\overline{T_2}$ relation is equivalent to $T_1<T_2$; the formal definition of $\mapsto$ is given later in Definition~\ref{definition:mapsto}.
We will also abbreviate $\langle k,v\rangle.\overline{T} \mapsto \langle k',v'\rangle.\overline{T}$ as simply $\langle k,v\rangle\mapsto \langle k',v'\rangle$.
The total order $\mapsto$ on timestamp vectors of values that stabilize is the total order on the linearization points of correct and pseudo-correct write operations.\footnote{As noted in Section~\ref{section:characterization}, the total order {\em in which} the values stabilize may be different from this total order because of concurrency races for stabilization between a pseudo-correct write $T_1(u_1)$ and a pseudo-correct or correct write $T_2(u_2)$, where $T_1<T_2$. If $T_2(u_2)$ stabilizes before $T_1(u_1)$, the write of $u_1$ will have a invisible linearization point as it will not be returned to any correct read operation.} This follows from the {\em no ``new-old'' inversions} clause in Theorem~\ref{theorem:lin}.


More than one value can stabilize 
as part of the same HLO write operation. This can happen when, for example, some of the readers' registers could have been written to in earlier Write operations or a HLO invocation-response of a Byzantine write contains multiple pseudo-correct write operations. 
For any pair of values that stabilize, this total order $\mapsto$ between them satisfies the Genuine Advance Property if $n>3t$, as we will prove in Theorem~\ref{theorem:advance}. 

With the introduction of the $\mapsto$ relation, the definition of register linearizability (Definition~\ref{definition:blin}) is adapted to the algorithm by rephrasing the {\tt No ``new-old'' inversions} property as follows.
\begin{definition}
\label{definition:lin}
(Byzantine Linearizabile Register for the algorithm). In a system with Byzantine process failures, an implementation of a SWMR register is linearizable if and only if the following two properties are satisfied.
\begin{itemize}
\item {\bf Reading a current value:} When a read operation {\tt R} by a non-Byzantine process returns the value $v$:
\begin{itemize}
    \item if $v=v_0$ then no correct or pseudo-correct write operation precedes {\tt R}
    \item else if $v\neq v_0$ then $v$ was written by the most recent correct write operation that precedes {\tt R} or by a later pseudo-correct or correct write operation (either a pseudo-correct write operation, that precedes or overlaps with {\tt R}, or a correct write operation that overlaps {\tt R}).
\end{itemize}
\item {\bf No ``new-old'' inversions:} If read operations {\tt R} and {\tt R'} by non-Byzantine processes return values $\langle k,v\rangle$ and $\langle k',v'\rangle$, respectively, and {\tt R} precedes {\tt R'}, then $\langle k,v\rangle.\overline{T} \stackrel{=}{\mapsto} \langle k',v'\rangle.\overline{T}$.
\end{itemize}
\end{definition}


\section{Correctness Proof}
\label{section:proof}
\begin{definition}
\label{definition:wsis0}
The witness set of an $Inform\_Set$ $IS$ that is formed, $WS(IS)$, is the maximal set of at least $n-t$ $T\_witness$ entries $\langle\langle k,u\rangle, s_l,l\rangle$ common to the at least $n-t$ elements ($Witness\_Set$s) of the $Inform\_Set$ $IS$.
\end{definition}
Those at least $n-t$ identical entries in the intersection of the at least $n-t$ $Witness\_Set$s in the $Inform\_Set$ $IS$ form $WS(IS)$.

\begin{definition}
\label{definition:stabilization}
The field/value $\langle k,u\rangle$ common to all the entries in the $WS(IS)$ witness set of an inform set $IS$ is defined to stabilize when the $Inform\_Set$ $IS$ containing correctly signed $Witness\_Set$s is written to all the $R\_final_{pi}$ for some 
process $p$.
\end{definition}

\begin{definition}
\label{definition:stabilizetime}
For a value $\langle k,u\rangle$ that stabilizes with $Inform\_Set$ $IS$, $\langle k,u\rangle.\overline{T}$ is the set of tuples $(l, s_l)$ for all the at least $n-t$ reader processes $l$ for all the at least $n-t$ $T\_witness_l$ 
entries $\langle\langle k,u\rangle, s_l, l\rangle$ in $WS(IS)$.
\end{definition}

Only a value that stabilizes may be returned by a reader.

\begin{theorem}
\label{theorem:correctstabilize}
A correct write operation is guaranteed to stabilize 
provided $n>2t$.
\end{theorem}
\begin{proof}
A correct write operation writes the same $\langle k,u\rangle$ to $R\_init_{wp}$ for all correct $p$ and will not complete unless it gets $n-t$ acks in $R\_ack_{pw}$. It may get $t$ acks from Byzantine processes but as $n>2t$, it will need an ack from at least one correct process. A correct process $p$ gives the ack only after it has written the value to $R\_final_{p*}$, i.e., when the value has stabilized. We now show that at least one correct process will have the value stabilize, before which the value written to the $R\_init_{w*}$ will not be overwritten.

For all correct $p$, the witness timestamps are correctly written to $R\_witness_{p*}$. At least $n-t$ correct reader helper threads of correct processes $p$ will eventually read from $R\_witness_{*p}$, form their $Witness\_Set$s, sign those sets and write to $R\_inform_{p*}$. Some first correct reader thread $p$ will eventually read from $R\_inform_{*p}$ and have at least $n-t$ $T\_inform_j$ having identical entries $\langle\langle k,u\rangle, s_l,l\rangle$ for at least $n-t$ processes $l$. It will thus be able to form its $Inform\_Set$, then write it to all $R\_final_{p*}$, and will then write $\langle k,u\rangle$ to $R\_ack_{pw}$ after which the correct write operation will complete. Thus, $\langle k,u\rangle$ is guaranteed to have stabilized as the $R\_init_{wi} (\forall i)$ will not be overwritten until then.
\end{proof}

\begin{theorem}
\label{theorem:lowerthreshold}
A potential pseudo-correct write operation may stabilize provided $n>2t$.
\end{theorem}
\begin{proof}
A value written by a 
potential pseudo-correct operation
writes the same value to $n-t$ readers' $R\_init_{wp}$, i.e., to at least $n-2t$ correct readers' $R\_init_{wp}$, across possibly multiple prior write operations and the current write operation. These reader processes write that value, if not overwritten, to $R\_witness_{*p}$. Let the $t$ Byzantine processes $b$ read the value from their $R\_witness_{*b}$ and thereafter behave as though the value had been written to their $R\_init_{wp}$ and thenceforth behave correctly. There is now a way that the value may stabilize if the Byzantine writer does not overwrite the values in $R\_init_{wp}$ ($\forall p$) until at least one 
process $q$ forms its $Inform\_Set$ of correctly signed $Witness\_Set$s for that value and writes the $Inform\_Set$ to $R\_final_{q*}$. This condition will be satisfied as per the logic in the proof of Theorem~\ref{theorem:correctstabilize} (although the writer need not wait for $n-t$ acks) as now the Byzantine reader processes $b$ are collaborating and behaving like correct reader processes after having read the value set aside for them in $R\_witness_{*b}$.
\end{proof}

\begin{theorem}
\label{theorem:onlystabilize}
If a value stabilizes, it must have been written by a correct write operation or by a potential pseudo-correct write operation. 
\end{theorem}
\begin{proof}
A correct write operation stabilizes (Theorem~\ref{theorem:correctstabilize}). So we need to only prove the following contrapositive, namely that if a write is not a potential pseudo-correct write operation, it will not stabilize.

If a write is not a potential pseudo-correct operation, it  is not written to at least $n-2t$ correct processes' $R\_init_{wi}$ across possibly multiple write operations. 
Then there is no way a $Witness\_Set$ of at least $n-t$ $T\_witness$ entries can form at at least $n-t$ 
processes, and hence an $Inform\_Set$ of $n-t$ correctly signed $Witness\_Set$s cannot form at any process, correct or Byzantine, and cannot be written to any $R\_final_{p*}$. 
If a Byzantine process attempts to write a fake $Inform\_Set$ in $R\_final_{p*}$, that will be detected by correct processes as that $Inform\_Set$ written in $R\_final_{p*}$ will not pass the signature test. Hence that value is deemed to not have stabilized.
\end{proof}

\begin{lemma}
\label{lemma:viewconsistency}
The View Consistency Property (Definition~\ref{definition:view}) is satisfied by Algorithm~\ref{alg:lin}.
\end{lemma}
\begin{proof}
If one correct process $p$ returns a value $v$ of write $T(v)$, it must have written the $Inform\_Set$ corresponding to this $T(v)$ to all $R\_final_{p*}$. If there are no further writes to $R\_init_{w*}$ beyond $T(v)$ by the writer, all correct readers' read operations issued after $p$ is returned $v$ will return $v$, based on the pseudo-code. Hence View Consistency is satisfied.
\end{proof}

If an $Inform\_Set$ $IS1$ is formed at reader $x$, it has read from $R\_inform_{*x}$ at least $n-t$ $Witness\_Set$s having at least $n-t$ identical entries across those $Witness\_Set$s. Likewise if an $Inform\_Set$ $IS2$ is formed at reader $y$. 
(Recall that those at least $n-t$ identical entries in the intersection of the at least $n-t$ $Witness\_Set$s in the $Inform\_Set$ $IS$ form $WS(IS)$.)
$x$ and $y$ write $Inform\_Set$ $IS1$ and $Inform\_Set$ $IS2$ to $R\_final_{x*}$ and $R\_final_{y*}$, respectively. When a third reader $z$ reads these two values and as part of $Find\_Latest(Z)$ invocation compares $IS1$ and $IS2$ ($Z$ and $Z'$), 
\begin{itemize}
\item there are at least $n-2t$ reader witness timestamps common to $WS1 (= WS(IS1))$ and $WS2 (= WS(IS2))$. As the corresponding processes provided witnesses to both witness sets, any such process would have done so first for $WS1$ and then for $WS2$ or vice-versa. 
\item there are at least $n-2t$ processes $p$ that provided (signed) $Witness\_Set$s written to $R\_inform_{p*}$ that formed part of both $IS1$ and $IS2$. Any such process $p$  would have 
written its $Witness\_Set$ that formed part of $IS1$ to $R\_Inform_{p*}$ before it wrote its $Witness\_Set$ that formed part of $IS2$ or vice-versa. 
\end{itemize}

\begin{observation}
    \label{nondecreasing}
    A correct process forms its successive $Witness\_Set$s only in non-decreasing order of source witness timestamps for the at least $n-2t$ common witnesses as it writes these $Witness\_Set$s to $R\_inform$. 
\end{observation}

Let $p$ sign $Witness\_Set$ $WS1'$ that is part of $IS1$. Then let $p$ sign $WS2'$ that is part of $IS2$. ($p$ is one of the at least $n-2t$ processes that sign $WS1' \in IS1$ and $WS2' \in IS2$, assuming $n>2t$. Note that these at least $n-2t$ processes may not include any correct process.) For all correct $p$, we have the property that all the at least $n-2t$ witnesses common to $IS1$ and $IS2$ have a higher or equal witness timestamp in $IS2$ than in $IS1$. 
Up to $t$ Byzantine processes $p$ can sign $WS1'$ that is part of $IS1$ and $WS2'$ that is intended to be part of $IS2$ such that some of the common witness timestamps common to $IS1$ and $IS2$ will have a smaller witness timestamp in $IS2$ than in $IS1$ (while some will have a greater or equal witness timestamp in $IS2$ than in $IS1)$. 
Such $IS1$ and $IS2$, which require a quorum of at least $n-t$ signed $Witness\_Set$s having identical $T\_witness$ entries $\langle\langle k,u\rangle, s_l,l\rangle$ for $\geq$ $n-t$ processes $l$ will form at any correct or Byzantine process only if $n-t\leq t$, ie., $n\leq 2t$. 
This is because the $t$ Byzantine processes $p$ would be using fake (out-of-order) witness timestamp entries for at least one of $WS1'$ and $WS2'$.
To prevent such a quorum $IS1$ or $IS2$ from forming, we require $n>2t$.

\begin{definition}
    \label{definition:mapsto}
    Given $WS1 (= WS(IS))$ and $WS2 (= WS(IS'))$, $WS1 \mapsto WS2$ iff for the at least $n-2t$ readers $z$ that witnessed values in both $WS1$ and $WS2$, $z$'s $WS1$ timestamp $\leq$ $z$'s $WS2$ timestamp and there is at least one reader $z'$ that witnessed values in both $WS1$ and $WS2$ and $z'$'s $WS1$ timestamp $<$ $z'$'s $WS2$ timestamp.
\end{definition}
If for any $z$ that witnessed both $WS1$ and $WS2$ the witness timestamps are equal, then $WS1$ and $WS2$ have the same $\langle k,u\rangle$ value. 

If $WS(IS)\mapsto WS(IS')$, we also interchangeably say that for the corresponding values, $\langle k,u\rangle \mapsto \langle k',u'\rangle$ and $\langle k,u\rangle.\overline{T} \mapsto \langle k',u'\rangle.\overline{T}$.

\begin{definition}
    \label{definition:conc}
    Given distinct $WS1 (= WS(IS))$ and $WS2 (= WS(IS'))$, $WS1 \| WS2$ iff $WS1 \not\mapsto WS2 \wedge WS2 \not\mapsto WS1$.
\end{definition}

\begin{observation}
    \label{observation:conc}
    If $WS1 || WS2$ then of the at least $n-2t$ processes that provided witness timestamps to both $WS1$ and $WS2$ there is a process $a$ whose $WS1$ witness timestamp $ta1$ $<$ its $WS2$ witness timestamp $ta2$ and there is a process $b$ whose $WS2$ witness timestamp $tb2$ $<$ its $WS1$ witness timestamp $tb1$.
\end{observation}

\begin{theorem}
    \label{theorem:notconc}
    For $Inform\_Set$s $IS1$ and $IS2$ at any two (possibly different) reader processes, $WS(IS1) \mapsto WS(IS2) \vee WS(IS2) \mapsto IS(IS1)$, i.e., $WS(IS1) \not\| \, WS(IS2)$, provided $n>2t$.
\end{theorem}
\begin{proof}
We prove by contradiction. Assume $WS(IS1) \| WS(IS2)$. Given $WS1 (= WS(IS1))$ and $WS2 (= WS(IS2))$, from Observation~\ref{observation:conc} there is a process $a$ whose $WS1$ witness timestamp $ta1$ $<$ its $WS2$ witness timestamp $ta2$ and there is a process $b$ whose $WS2$ witness timestamp $tb2$ $<$ its $WS1$ witness timestamp $tb1$. As noted earlier, a process $p$ that provided witness set inputs to both $IS1$ and $IS2$ by writing to $R\_inform_{p*}$ could have provided its input for $IS1$ first and then for $IS2$ or vice-versa. 

Without loss of generality assume $p$ provides the witness set input for $IS1$ before providing the witness set input for $IS2$. Then consider the process $b$'s witness timestamps. So let 
process $p$ provide $b$'s $WS1$ witness timestamp $tb1$ (along with other $WS1$ witness timestamps) which is written to $R\_inform_{p*}$. If $p$ is correct, it will not consider/input $b$'s $WS2$ witness timestamp $tb2$ as $tb2< tb1$. Only up to $t$ Byzantine processes can process/input $b$'s $WS2$ timestamp $tb2$ after $tb1$ but, as $n>2t$, that falls short of the $n-t$ threshold required to form an $Inform\_Set$ (of signed $Witness\_Set$s) at any process. So $IS2$ and $WS2$ will not exist.

Likewise if we assume the process $p$ provides its witness set input for $IS2$ before providing its witness set input for $IS1$, then consider process $a$'s witness timestamps. Let 
process $p$ provide $a$'s $WS2$ witness timestamp $ta2$ (along with other $WS2$ witness timestamps) which is written to $R\_inform_{p*}$. If $p$ is correct, it will not consider/input $a$'s $WS1$ witness timestamp $ta1$ as $ta1< ta2$. Only up to $t$ Byzantine processes can process/input $a$'s $WS1$ witness timestamp $ta1$ after $ta2$. But as $n>2t$, each correct or Byzantine process will fall short of the $n-t$ threshold required to form an $Inform\_Set$ (of signed $Witness\_Set$s). So $IS1$ and $WS1$ will not exist.

Thus if $IS1$ and $IS2$ form, all processes that provide input witness timestamps to both $WS1$ and $WS2$ provide first to $WS1$ and then to $WS2$ or all provide first to $WS2$ and then to $WS1$. Thus $WS1 \not\| WS2$.
%
\end{proof}

\begin{definition}
\label{definition:advance}
({\em Genuine Advance Property in the algorithm:})
When $WS1 (=WS(IS1))\mapsto WS2 (=WS(IS2)$, $WS2$ is a {\em genuine advance} over $WS1$ if there is at least one correct reader process $i$ such that $i$'s $WS1$ witness timestamp $<$ $i$'s $WS2$ witness timestamp. 
%
\end{definition}
Only value $\langle k,v\rangle$ corresponding to an $Inform\_Set$ that is written to $R\_final_{p*}$ could be returned by a correct process if the signed $Witness\_Set$s in that $Inform\_Set$ pass the signature test.
The Genuine Advance Property is useful because each new value returned by a correct reader is a new value genuinely written to at least one correct reader $i$'s $R\_init_{wi}$ and not falsely reported by a Byzantine reader, in addition to that same value being the most recent value in a total of $n-t$ readers $j$'s $R\_init_{wj}$ registers as per $\langle k,u\rangle.\overline{T}$. Another very important reason why this property is important is explained after Theorem~\ref{theorem:monocor}. Definition~\ref{definition:advance} is the counterpart of Definition~\ref{definition:advancetv}.

\begin{theorem}
\label{theorem:advance}
Algorithm~\ref{alg:lin} satisfies the Genuine Advance Property if $n>3t$.
\end{theorem}
\begin{proof}
If $WS1 (=WS(IS1)) \mapsto WS2 (=WS(IS2))$, then in order that the Genuine Advance Property holds, there is at least one correct reader process $i$ such that $i$'s $WS1$ witness timestamp $<$ $i$'s $WS2$ witness timestamp. 
For this to happen, we require that two quora of size $n-2t$, which is the minimum number of correct processes having $T\_witness$ entries in $WS1$ and $WS2$, intersect among the set of correct processes (having size $\geq n-t$). Thus,
\[2(n-2t) > n-t \Longrightarrow n> 3t\]
\end{proof}

Only value $\langle k,v\rangle$ corresponding to an $Inform\_Set$ that is written to $R\_final_{p*}$ could be returned by a correct process if the signed $Witness\_Set$s in that $Inform\_Set$ pass the signature test.
From Theorem~\ref{theorem:notconc}, all the $\langle k,v\rangle.\overline{T}$ form a total order based on $\mapsto$. This leads to the following corollary.

\begin{corollary}
\label{corollary:mapstoto}
The partial vector timestamps $\langle k,u\rangle.\overline{T}$ of $\langle k,u\rangle$ values that stabilize are totally ordered.
\end{corollary}
Recall that the $\langle k,u\rangle.\overline{T}$ are partial vectors having $\geq n-t$ entries because not all $n$ readers' entries may be present in the $Witness\_Set$s of the $Inform\_Set$. Let those entries that are not reported have a timestamp value denoted $\perp$ in $\overline{T}$.

\begin{theorem}
\label{theorem:monocor}
The Monotonicity/Total Order of Vector Timestamps of Stabilized Writes Property (Definition~\ref{definition:monotonicity}) and the Genuine Advance Property (Definition~\ref{definition:advancetv})
are satisfied by Algorithm~\ref{alg:lin}, provided $n>2t$ and $n>3t$, respectively.
\end{theorem}
\begin{proof}
For every partial timestamp vector $\overline{T}$ of a value that stabilizes in the algorithm, we construct a full timestamp vector $T$ as follows.
Let $\overline{T}_{current}$ and $T_{current}$ be the corresponding timestamps of the current value that has stabilized. 
Let $\overline{T}_{next}$ be the smallest timestamp greater than $\overline{T}_{current}$ to stabilize and let $T_{next}$ denote the corresponding full vector timestamp we construct. 
\medskip

\noindent{Initialize:} $\overline{T}_{current}, T_{current}\leftarrow [0,\ldots0]$\\
\underline{loop:}\\
\hspace*{0.5cm} Identify $\overline{T}_{next}$\\
\hspace*{0.5cm} // {\bf Invariant 1:} $\overline{T}_{current}\mapsto\overline{T}_{next}$\\
\hspace*{0.5cm} $\forall i\,|\,\overline{T}_{next}[i]\neq\perp: T_{next}[i]\leftarrow \overline{T}_{next}[i]$\\ 
\hspace*{0.5cm} $\forall i\,|\,\overline{T}_{next}[i]=\perp: T_{next}[i]\leftarrow T_{current}[i]$\\
\hspace*{0.5cm} // {\bf Invariant 2:} $T_{current}<T_{next}$\\
\hspace*{0.5cm} $T_{current}\leftarrow T_{next}$, $\overline{T}_{current}\leftarrow \overline{T}_{next}$\\
\underline{endloop}

\medskip

Let $\overline{{\cal T}}$ denote the (total order) set of $\langle k,u\rangle.\overline{T}$ partial vector timestamps for values $\langle k,u\rangle$ that have stabilized. Let ${\cal T}^f$ denote the set of corresponding full vector timestamps. 

\begin{theorem}
\label{theorem:isomorph}
$({\overline{\cal T}},\mapsto)$ is isomorphic to $({\cal T}^f,<)$.
\end{theorem}

\begin{proof}
With respect to Invariant 1, 
\begin{itemize}
\item let $j$ be any index such that $\overline{T}_{current}[j]<\overline{T}_{next}[j]$, 
\item let $k$ be any index such that $\overline{T}_{current}[k]=\overline{T}_{next}[k]\neq\perp$, 
\item let $l$ be any index such that  $\overline{T}_{current}[l]=\overline{T}_{next}[l]=\perp$, 
\item let $a$ be any index such that  $\overline{T}_{current}[a]\neq\perp$ and $\overline{T}_{next}[a]=\perp$,
\item let $b$ be any index such that $\overline{T}_{current}[b]=\perp$ and $\overline{T}_{next}[b]\neq\perp$.
\end{itemize}
Invariant 2 follows because of the following.
\begin{itemize}
\item For all $j$, $T_{current}[j]<T_{next}[j]$,
\item for all $k$, $T_{current}[k]=T_{next}[k]$,
\item for all $l$, $T_{current}[l]=T_{next}[l]$,
\item for all $a$, $T_{current}[a]=T_{next}[a]$,
\item for all $b$, $T_{current}[b]\leq T_{next}[b]$.
\end{itemize}
From Invariant 1, there must exist at least one index $j$, or one index $b$ such that $T_{current}[b]<T_{next}[b]$. Hence $T_{current}<T_{next}$ and Invariant 2 holds.

In order to satisfy the Genuine Advance of Timestamps Property (Definition~\ref{definition:advancetv}) for $T_{current}<T_{next}$, there must exist a {\em correct} process index $j$, or $b$ satisfying $T_{current}[b]<T_{next}[b]$, in the analysis above. This requires 2 quora of $n-2t$ processes to intersect among the set of correct processes, requiring $n>3t$ as shown in the proof of Theorem~\ref{theorem:advance}.

The theorem follows from Invariants 1 and 2.
\end{proof}
As $(\overline{{\cal T}},\mapsto)$ is a total order from the proof of Theorem~\ref{theorem:notconc}, Theorem~\ref{theorem:isomorph} implies that $({\cal T}^f,<)$ is also a total order. The Monotonocity/Total Order of Vector Timestamps of Stabilized Writes Property is thus satisfied, and requires $n>2t$ as Theorem~\ref{theorem:notconc} also requires it.

From the proof of Theorem~\ref{theorem:isomorph},
the Genuine Advance of Timestamps Property over $({\cal T}^f,<)$ is satisfied by the algorithm, provided $n>3t$.
\end{proof}

Note that when $3t\geq n>2t$, $n-t$ different values written to different correct readers' $R\_init$ registers can be ordered/stabilized in any permutation by the Byzantine processes but in any given execution, only one permutation can occur. But as the Genuine Advance Property is not satisfied when $n\leq 3t$ as is this case, the Byzantine readers can cause any arbitrary sequence of unbounded length (with each member of the sequence being distinct from the one before it) of these $n-t$ different values to successively stabilize and be returned by correct readers. This is an unbounded sequence of fake writes that can be returned to reads. This is another reason why the Genuine Advance Property is important.

An $Inform\_Set$ that is formed at any process $p$ is written to $R\_final_{p*}$ and thus stabilizes, by Definition~\ref{definition:stabilization}. We now have the following corollary to Theorems~\ref{theorem:notconc} and ~\ref{theorem:advance}. 

\begin{corollary}
    \label{corollary:totalorder}
    The set of all values that stabilize is totally ordered by $\mapsto$ provided $n>2t$ (from Theorem~\ref{theorem:notconc}). The Genuine Advance Property is satisfied provided $n>3t$ (from Theorem~\ref{theorem:advance}). 
\end{corollary}

\begin{theorem}
    \label{theorem:lin}
    Algorithm~\ref{alg:lin} implements a linearizable SWMR register using SWSR registers, provided $n>3t$.
\end{theorem}
\begin{proof}
We show that the value returned by a read operation satisfies ``reading a current value'' and ``no new-old inversions''.
\begin{itemize}
\item {\bf Reading a current value:}
From the algorithm pseudo-code, a read {\tt R} returns the value $\langle k,u\rangle$ such that $\langle k,u\rangle.\overline{T}$ is the latest value ordered by $\mapsto$ that has stabilized, up until some point of time during {\tt R}. By Theorem~\ref{theorem:onlystabilize}, such a value must have been written by a correct write operation or by a pseudo-correct write operation that 
has stabilized.
From Theorem~\ref{theorem:correctstabilize}, a correct write always stabilizes before the write operation returns/completes. Therefore, $\langle k,u\rangle$ will be the value written by the most recent correct write operation that precedes {\tt R}, or by a later write operation. The later write operation may be (i) a correct write operation that overlaps with {\tt R}, or (ii) a pseudo-correct write operation that has stabilized (and which precedes or overlaps {\tt R}). Thus a current value is read.
\item {\bf No ``new-old'' inversions:} Let read {\tt R} by $i$ return $\langle k,u\rangle$ and let read {\tt R'} by $j$ return $\langle k',u'\rangle$, where {\tt R} precedes {\tt R'}. {\tt R} will write $\langle k,u\rangle$ in $R\_final_{i*}$ before returning. {\tt R'} will read from $R\_final_{*j}$, add these elements to $Z$ and invoke $Find\_Latest(Z)$ which will return the most recent value $\langle k,u\rangle^{max}$ as per $\mapsto$. It is guaranteed that if $\langle k,u\rangle^{max}\neq \langle k,u\rangle$ then $\langle k,u\rangle.\overline{T}\mapsto \langle k,u\rangle^{max}.\overline{T}$ because all the values in $Z$ are totally ordered by $\mapsto$ (Theorem~\ref{theorem:notconc}, Corollary~\ref{corollary:totalorder}) and from Definition~\ref{definition:mapsto}, $\langle k,u\rangle.\overline{T}\mapsto \langle k,u\rangle^{max}.\overline{T}$ implies that the write of $\langle k,u\rangle^{max}$ had greater (or equal) witness timestamps than the write of $\langle k,u\rangle$ for the at least $n-2t$ common processes that witnessed both writes.
The value $\langle k',u'\rangle$ returned by {\tt R'} is $\langle k,u\rangle^{max}$. Thus there are no inversions.
\end{itemize}
The Genuine Advance Property implicitly needed requires $n>3t$ (Theorem~\ref{theorem:advance}).
\end{proof}

Let correct reader $i$ return $\langle k,u\rangle$ to {\tt R(i,1)} and return $\langle k',u'\rangle$ to {\tt R(i,2)}, where {\tt R(i,1)} precedes {\tt R(i,2)}. Let correct reader $j$ return $\langle k',u'\rangle$ to {\tt R(j,1)} and let it then issue {\tt R(j,2)}, where {\tt R(j,1)} precedes {\tt R(j,2)}. As the register is Byzantine linearizable (Theorem~\ref{theorem:lin}), from the values returned to $i$, we have $\langle k,u\rangle.\overline{T}\mapsto \langle k',u'\rangle.\overline{T}$. If {\tt R(j,2)} were to return $\langle k,u\rangle$, $\langle k',u'\rangle.\overline{T}\mapsto \langle k,u\rangle.\overline{T}$ which leads to a contradiction. Hence {\tt R(j,2)} must return $\langle k',u'\rangle$ or a value with a higher timestamp vector $\overline{T}$. Thus, the Total Ordering Property cannot be violated by the algorithm. 
This logic along with Theorem~\ref{theorem:advance} about the Genuine Advance Property, implicitly needed for linearization, gives the following corollary.

\begin{corollary}
\label{corollary:tocor}
The Total Ordering property (Definition~\ref{definition:to}) of values returned by correct readers is satisfied, provided $n>3t$.
\end{corollary}

As stated in Section~\ref{section:modelprelim}, a SWMR register supports Byzantine linearizable executions because before every read operation of a correct process, one can add a corresponding Byzantine write in the linearization \cite{DBLP:conf/wdag/CohenK21}. Next, we elaborate on this to give a linearization of an execution to show that the SWMR Byzantine linearizable register we constructed supports Byzantine linearizability of executions. First, read operations of correct readers are linearized in the order of the return values, as per the $\mapsto$ relation. Then an update of a value $\langle k,v\rangle$ by the Byzantine writer is added just before the first read operation that reads that value. In general, the value returned by a read operation is based on the values written by one or more than one HLI write operation as the writer is Byzantine. Further, one HLI write operation by the Byzantine writer may result in 
different read values
being returned by multiple correct readers. To differentiate among the multiple updates of different values over time to the SWMR register, 
we treat each update, 
which has stabilized, as an independent {\em Byzantine (correct or pseudo-correct) write operation}, associated with its vector timestamp. 

\begin{theorem}
\label{theorem:blin}
The SWMR Byzantine linearizable register implemented by Algorithm~\ref{alg:lin} satisfies Byzantine linearizability of executions.
\end{theorem}
\begin{proof}
Let $L_R$ be a linearization of the correct readers' read operations such that (i) the local order of reads at each such reader is preserved, and (ii) if read {\tt R} returns $\langle k,v\rangle$, read {\tt R'} returns $\langle k',v'\rangle$, and $\langle k,v\rangle.\overline{T}$ $\mapsto$ $\langle k',v'\rangle.\overline{T}$ then {\tt R} precedes {\tt R'} in $L_R$.

When $\langle k,v\rangle$ is returned by a read, the latest write operation to which any of the witness timestamps in $\langle k,v\rangle.\overline{T}$ belongs is denoted as $latest\_W({\tt R})$.  In the linearization $L_R$, place a (Byzantine) write operation writing a value $\langle k,v\rangle$ and assigned timestamp $\langle k,v\rangle.\overline{T}$ immediately before the first read operation {\tt R} in $L_R$ that reads $\langle k,v\rangle$ -- this is one of the possibly multiple Byzantine write operations corresponding to the write operation $latest\_W({\tt R})$.

The resulting linearization is seen to be a Byzantine linearization that considers all the read operations of correct readers, and includes some Byzantine write operations. In the extreme case, before every correct read operation in $L_R$, we add a corresponding Byzantine write operation.
\end{proof}

Additionally, the algorithm satisfies Monotonicity/Total Order of Stabilized Vector Timestamps (Def.~\ref{definition:monotonicity} \& Theorem~\ref{theorem:monocor}, $n>2t$), Genuine Advance of Timestamps (Def.~\ref{definition:advancetv} \& Theorem~\ref{theorem:monocor}, $n>3t$), View Consistency (Def.~\ref{definition:view} \& Lemma~\ref{lemma:viewconsistency}, $n>2t$), Total Ordering (Def.~\ref{definition:to} \& Corollary~\ref{corollary:tocor}, $n>3t$).

\subsection*{Space Complexity}
The algorithm uses $3n^2$ shared SWSR registers: $n^2$ $R\_witness$ registers of size $O(1)$, $n^2$ $R\_inform$ registers of size $O(n)$, $n^2$ $R\_final$ registers of size $O(n^2)$. It also uses $2n^2$ shared SWSR registers: $n$ $R\_init$ registers of size $O(1)$, and $n$ $R\_ack$ registers of size $O(1)$.

The local space at each reader process can be seen to be $O(n^2)$. 

\section{Conclusions}
\label{section:conc}
This paper studied Byzantine tolerant construction of a SWMR atomic register from SWSR atomic registers. It is the first to propose a definition of Byzantine register linearizability by non-trivially taking into account Byzantine behavior of the writer and readers, and by overcoming the drawbacks of the definition used by previous works. 
We introduced the concept of a correct write operation by a Byzantine writer. We also introduced the notion of a pseudo-correct write operation by a Byzantine writer, which has the effect of a correct write operation. Only correct and pseudo-correct writes may be returned by correct readers. The correct and pseudo-correct writes are totally ordered by their linearization points and this order is the total order in logical time in which the writes were performed. 
We then gave an algorithm to construct a Byzantine tolerant SWMR atomic register from SWSR atomic registers that meets our definition of Byzantine register linearizability.

\bibliographystyle{plain}
\bibliography{references}

\end{document}